\documentclass{amsart}

\usepackage[latin1]{inputenc}
\usepackage[all]{xy}
\usepackage{amssymb, amsmath, amscd, geometry}
\usepackage{graphicx}
\usepackage{setspace}

\numberwithin{equation}{section}
\input xy
\xyoption{all}

\usepackage{latexsym}
\usepackage[usenames]{color}
\usepackage{epic} 
\usepackage{mathrsfs}
\usepackage{euscript}
\usepackage{rotating}
\usepackage{verbatim}

\newtheorem{lemma}{Lemma}[section]
\newtheorem{theorem}[lemma]{Theorem}
\newtheorem{corollary}[lemma]{Corollary}

\newtheorem{prop}[lemma]{Proposition}

\newtheorem*{theorem*}{Theorem}

\theoremstyle{definition}
\newtheorem{definition}{Definition}[section]
\newtheorem{remark}{Remark}[section]

\newcommand{\K}{{\mathbb K}}      \newcommand{\N}{{\mathbb N}}
      \newcommand{\R}{{\mathbb R}}
      \newcommand{\C}{{\mathbb C}}

\newcommand{\B}{{\mathcal B}}

\newcommand{\h}{\mathcal H}
\newcommand{\e}{\mathcal E}

\newcommand{\tr}{\operatorname{tr}}
              \newcommand{\Id}{\operatorname{Id}}

\newcommand{\beq}{\begin{eqnarray}}
\newcommand{\eeq}{\end{eqnarray}}

\newcommand{\setft}[1]{\mathrm{#1}}

\newcommand{\Lin}{\setft{L}}
\newcommand{\Bounded}{\mathcal{B}}
\newcommand{\Ball}{\setft{Ball}}

\newcommand{\uno}{1\!\!1}

\title[On the power of quantum entanglement in multipartite quantum XOR games]{On the power of quantum entanglement  in multipartite quantum XOR games}

\author{Marius Junge}
\address{Department of Mathematics\\University of Illinois at Urbana-Champaign\\1409 W. Green St. Urbana, IL 61891. USA}
\email{junge@math.uiuc.edu}

\author{Carlos Palazuelos}
\address{Instituto de Ciencias Matem\'aticas (ICMAT)\\Departamento de An\'alisis Matem\'atico y Matem\'atica Aplicada \\
Facultad de Ciencias Matem\'aticas \\ Universidad Complutense de Madrid \\
Madrid 28040. Spain}
\email{cpalazue@ucm.es}

\thanks{MJ is partially supported by the NSF grants DMS 1800872 and Raise-TAG183917. C.P. is partially supported by the MICINN project PID2020-113523GB-I00, by the grant Grant CEX2019-000904-S funded by MCINN/AEI/ 10.13039/501100011033 and by QUITEMAD+-CM, P2018/TCS4342, funded by Comunidad de Madrid}

\addtocounter{tocdepth}{-1}

\begin{document}

\addtolength{\parskip}{+1ex}

\keywords{Nonlocal games, Quantum XOR games, Operator spaces}

\maketitle

\begin{abstract}
In this paper we show that, given $k\geq 3$, there exist $k$-player quantum XOR games for which the entangled bias can be arbitrarily larger than the bias of the game when the players are restricted to separable strategies. In particular, quantum entanglement can be a much more powerful resource than local operations and classical communication to play these games. This result shows a strong contrast to the bipartite case, where it was recently proved that the entangled bias is always upper bounded by a universal constant times the one-way classical communication bias.
\end{abstract}

\section{Introduction and main results}

One of the aims of research in quantum information is to quantify the advantages arising from the use of quantum resources as opposed to the use of classical resources. It is known that quantum computers can solve certain problems exponentially faster than a classical computer and also that computing a function of two people's inputs can be done with exponentially less communication with quantum messages than with classical ones. 

Another prominent context where the problem of comparing quantum and classical resources has been deeply studied is that of nonlocal games. In such games, $k$ spacelike separated parties (players) $A_1,\cdots, A_k$ receive inputs $x_1,\cdots ,x_k$ from a verifier according to some fixed and known probability distribution $\pi$, and are required to produce outputs $a_1,\cdots, a_k$. There is a predicate function $V(a_1,\cdots, a_k|x_1,\cdots , x_k)$ specifying which output $a_1,\cdots, a_k$ are considered ``winning'' on inputs $x_1,\cdots , x_k$ while the others are ``losing''. The value of such a game is defined as the maximum success probability that the players can achieve, which obviously depends on the type of strategies the players are allowed to use.  

A very interesting (and arguably the simplest) type of games are XOR games. These are nonlocal games in which each player only provides the referee with a one-bit answer and the predicate function depends on the parity of the $k$ bits alone. Despite its simplicity these games are reach enough to show very interesting phenomena. In particular, there exist bipartite XOR games for which an entangled state shared by the players allows them to define strategies leading to a success probability strictly greater than what can be achieved by unentangled players \cite{CHSH69, CHTW04}. This phenomenon is known as ``Bell inequality violations''  in the quantum information literature and it has been key in the study of the foundations of quantum mechanics \cite{Bell64}  as well as in some of the most striking applications of quantum information such as device independent quantum cryptography \cite{BHK05, VaVi14} and random number generators \cite{Colbeck14, Pironio09}. However, the interest of nonlocal games goes beyond quantum information theory, since they are also crucial in complexity theory \cite{CHTW04, ItVi12, JNVWY, NaWr} and they have been also used very recently to study certain problems in functional analysis \cite{Fritz12, JNPPSW11, PaVi16}.

On the other hand, in order to test the power of entanglement compared to classical resources, one could seek for a more extreme behavior than that mentioned in the previous paragraph; namely some contexts where the use of entanglement provides some advantages with respect to the interchange of classical communication between the players.  At the end, shouldn't classical communication be considered a classical resource? Nevertheless, in this case one should immediately note that nonlocal games, in the way presented above, cannot be used for this purpose. Indeed, if we think of a two-player (Alice and Bob) game, where the players are allowed to interchange classical communication, after receiving her question Alice could send it to Bob as part of their strategy, allowing the players to act as a single person with access to both questions. It is then clear that no ``entangled strategy'' (nor any other type of strategy satisfying some basic physical requirements) could ever improve what the players can do with classical communication. On the other hand, one can easily imagine a purely quantum task to be solved (for instance, one where entanglement must be created) for which the only use of classical communication is completely useless, yet the use of an entanglement state allows the players to solve the task in a trivial way.  

There are still many natural contexts where the comparison between entanglement and classical communication becomes very relevant and highly non trivial.  One interesting problem is the ``distinguishability of quantum states and data hiding'' \cite{DLT02, EgWe02, MWW09},  where different parties $A_1,\cdots , A_k$ receive one among several $k$-partite states and they have to distinguish which one they got. We will study this problem from the point of view of the so called quantum XOR games, which were introduced in \cite{ReVi15}. In complete analogy with their classical counterpart,  a $k$-player quantum XOR game is described by means of a family of $k$-partite quantum states $(\rho_x)_{x=1}^N$, a family of signs $c=(c_x)_{x=1}^N\in \{-1,1\}^N$ and a probability distribution $p=(p_x)_{x=1}^N$ on $\{1,\cdots, N\}$.  As before, the game starts with the referee choosing one of the states $\rho_x$ (which will now play the role of a quantum question) according to the probability distribution $p$. Then, the referee sends register $H_{A_i}$ to the player $A_i$ for every $i=1,\cdots ,k$. After receiving the states, the players must answer an output, $a_i=\pm 1$ in the case of $A_i$. Then, the players win the game if $a_1\cdots a_k=c_x$. Quantum XOR games can be understood as a particular situation of the distinguishability of quantum states, where the parties are restricted to certain conditions (see \cite[Section 4]{CLP14} for details about the precise relation between both problems).  
 
Since the players receive now quantum states it is not a priori clear that classical communication can be of any help.  One could  then guess the existence of quantum XOR games for which the use of entanglement is arbitrarily ``better'' than the use of classical communication. More precisely, if $\beta^*(G)$ denotes the entangled bias\footnote{For some reasons that will become clear in Section \ref{Sec: correlations}, when working with XOR games one usually works with the bias $\beta=2P_{win}- 1$ rather than with the winning probability $P_{win}$.} of the game $G$ and  $\beta_{LOCC}(G)$ denotes the bias of the game when the players are allowed to interchange classical communication and perform local operations as part of their strategies (local operations should always be allowed to the players), one could wonder whether there exist quantum XOR games $G$ for which $\beta^*(G)$ is much greater than $\beta_{LOCC}(G)$. This question was recently studied in \cite{JKPV22} and it was proved that, somehow surprisingly, this cannot happen for bipartite games. That is, there exists a universal constant $K$ such that $\beta^*(G)\leq K \beta_{LOCC}(G)$ for any bipartite quantum XOR games, independently of the size of $G$. This result was obtained by first relating the quantities $\beta^*(G)$ and $\beta_{LOCC}(G)$ with certain norms defined on linear maps between operator spaces and, then, applying a noncommutative version of Grothendieck theorem.

The main result of this work shows that the situation drastically changes for games with more than two players. Indeed, we will show that in the multipartite setting, there exist quantum XOR games $G$ for which $\beta^*(G)$ can be arbitrarily larger than $\beta_{LOCC}(G)$. In fact, we will show the following stronger result:
\begin{theorem}\label{Main Thm}
For every $k\geq 3$ there exists a constant $C(k)$ depending only on $k$ with the following property: Given natural numbers $D$ and $m$ satisfying that $D\geq c m^4 \sqrt{\log m}$ (where $c$ is a universal constant), there exists a $k$-player quantum XOR game with local dimension $D$ such that $$\beta^*(G)\geq C(k)m^{\frac{k}{2}-1}(\log m)^{-\frac{k}{2}}\beta_{SEP}(G),$$where $\beta_{SEP}(G)$ denotes the bias of the game when the players are restricted to the use of the so called separable strategies.
\end{theorem}

Since it is well known that every LOCC strategy is separable \cite{CLMOW14} (so, it holds that $\beta_{LOCC}(G)\leq \beta_{SEP}(G)$ for every game $G$), the previous result guarantees the existence of $k$-player quantum XOR games (for $k\geq 3$) for which the rato $\beta^*(G)/\beta_{LOCC}(G)$ can be arbitrarily large by just increasing $m$, so the size of the game $D$. Moreover, although the order $m^{\frac{k}{2}-1}(\log m)^{-\frac{k}{2}}$ is probably not optimal, it must be stressed that the scaling of the ratio $\beta^*(G)/\beta_{SEP}(G)$ given by Theorem \ref{Main Thm} is polynomial in the size of the game $D$, something needed for potential experimental realizations.

From a more general point of view, Theorem \ref{Main Thm} gives us new insights about the power of multipartite quantum entanglement. Indeed, while bipartite entanglement has been well researched and understood in the past decades, much less is known about the multipartite case. Our main result presents an extreme behavior of multipartite quantum entanglement showing that, even in very simple scenarios as it is that of quantum XOR games, it can be more powerful than the use of classical communication (independently of the amount of communication and the way this communication is interchanged). In fact, together with the main result in \cite{JKPV22} mentioned above which states that this cannot happen for bipartite quantum XOR games, Theorem \ref{Main Thm} shows that, when comparing with classical resources, quantum entanglement becomes a more extreme resource in the multipartite setting. There is here an analogy with the case of classical XOR games. Indeed, if $G$ is a bipartite classical XOR game and we denote by $\beta(G)$ and  $\beta^*(G)$ its unentangled bias\footnote{Note that here the only resource allowed to the players is the use of shared randomness and any interchange of  classical communication between them is forbidden.} and its entangled bias respectively, Tsirelson showed \cite{Tsirelson} that, according to the classical Grothendieck inequality, one has $\beta^*(G)\leq K_G^{\R}\beta(G),$ where $K_G^{\R}$ denotes the real Grothendieck constant. However, it was proved later, in \cite{PWPVJ08}, that for tripartite classical XOR games the ratio $\beta^*(G)/\beta(G)$ can be made arbitrarily large. Thus, Theorem \ref{Main Thm} can be understood as an analogous result to the main result in \cite{PWPVJ08}, when now entanglement improves, not just the use of shared randomness, but the much more powerful resource of classical communication. As we explained above, the prize to pay to show this phenomenon is that one must deal with non-classical games.

Finally, Theorem \ref{Main Thm} can also be understood as a counterexample of an extension of a certain noncommutative Grothendieck theorem. Indeed, the main result in \cite{JKPV22} shows that there exists a universal constant  $K$ such that $\beta^*(G)\leq K \beta_{LOCC}(G)$ for any bipartite quantum XOR game, independently of the size of $G$. In fact, the result presented in \cite{JKPV22} is even stronger, since there $\beta_{LOCC}(G)$ can be replaced by the one-way version of it, $\beta_{LOCC \rightarrow}(G)$, which satisfies $\beta_{LOCC \rightarrow}(G)\leq \beta_{LOCC}(G)$ for every game $G$. The main technical tool to prove that result is a noncommutative version of Grothendieck theorem for bilinear forms. Theorem \ref{Main Thm} implies that a similar theorem cannot hold for $k$-multilinear forms (see Remark \ref{explanation NCGT} for a detailed explanation about this point). 

The paper is organized as follows: In Section \ref{Sec: Preliminaries} we introduce some basic concepts about Banach spaces and operator spaces that we will use along the paper. Section \ref{Sec: correlations} is devoted to the introduction of quantum XOR games, as well as to developing the mathematical description of the different bias of these games depending on the resources allowed to the players. Finally, in Section \ref{Sec: Proofs} we prove the main results of the paper.

\section{Preliminaries}\label{Sec: Preliminaries}
In this section we will introduce some concepts about Banach spaces and operator spaces which will be used in this work. 

Let us start fixing some notation. Given linear spaces $X$ and $Y$, let $\Lin(X,Y)$ denote the set of linear maps from $X$ to $Y$ and let $X^\star=\Lin(X,\K)$ denote the algebraic dual space of $X$. When $X$ and $Y$ are Banach spaces, we denote by $\Bounded(X,Y)$ the Banach space of bounded linear maps from $X$ to $Y$, by $X^*=\Bounded(X,\K)$ the dual Banach space of $X$ and by $\Ball(X)$ the unit ball of $X$. In the particular case $X=Y$, we write $\Bounded(X)$. We also write $\Id_X\in \Bounded(X)$ for the identity map on $X$. In this work, the space $\C^n$  is always endowed with the Hilbertian norm and identified with $\ell_2^n$. The inner product between two elements $x,y\in\ell_2^n$ is denoted by $\langle x| y\rangle$. Given any natural number $N$, we denote $[N]=\{1,\cdots, N\}$.

Given a linear space $X$, $M_d(X)$ is the set of $d\times d$ matrices with entries in $X$; algebraically, $M_d(X)=M_d\otimes X$. When we simply write $M_d$ we always mean $M_d(\C)$. We write $\{e_{ij}\}$ for the canonical basis of $M_d$.

Given a complex Hilbert space $\h$, $S_p(\h)$ is the Schatten $p$-space for $1\leq p\leq \infty$ and $S_p^{\mathrm{sa}}(\h)$ is the space of self-adjoint operators which are in $S_p(\h)$. Recall that $S_p(\h)^*=S_q(\h)$ isometrically whenever $p,q\in (1,\infty)$ and $1/p+1/q=1$, where the dual action is given by $\langle A, B\rangle=tr(A^TB)$ for every $A\in S_p(\h)$ and $B\in S_q(\h)$.  We also have $S_\infty(\h)^*=S_1(\h)$ and $S_1(\h)^*=\mathcal B(\h)$ isometrically. In particular, if $\h$ has finite dimension, $S_p(\h)^*=S_q(\h)$ holds isometrically for every conjugate $p,q\in [1,\infty]$. When $\h$ is a finite-dimensional Hilbert space, we write $\uno_{\B(\h)}$ to denote the identity operator on $S_p(\h)$ for every $p$, since all the spaces $S_p(\h)$ and $\B(\h)$ are isomorphic.

We introduce some notions about tensor norms~\cite{Def93}. Given $k$ linear spaces $X_1,\cdots, X_k$, we denote by $X_1\otimes\cdots \otimes X_k$ its algebraic tensor product. If $X$ and $Y$ have finite dimension, one obtains the natural identification 
\begin{align}\label{eq:lin-tensor}
L(X,Y)=X^\star\otimes Y.
\end{align}

To see Eq. (\ref{eq:lin-tensor}), given a dual basis $\big((e_i)_i,\, (e^*_i)_i\big)$ of $\big(X,\, X^*\big)$, for any linear map $T:X\rightarrow Y$ we associate the tensor $\hat{T}=\sum_{i}e_i^*\otimes T(e_i)\in X^\star\otimes Y$. On the other hand, given a tensor $z=\sum_i x_i^*\otimes y_i\in X^\star\otimes Y$, the corresponding linear map is defined as $T_z(x)=\sum_i x_i^*(x) y_i$ for every $x\in X$.

The identification~(\ref{eq:lin-tensor}) can be made isometric. First, given $k$ Banach spaces $X_1,\cdots , X_k$, define the \emph{injective tensor norm} of $z\in X_1\otimes\cdots \otimes X_k$ as
\begin{align}\label{Def injective norm}
\|z\|_{X_1\otimes_\epsilon\cdots \otimes_\epsilon X_k}=\sup\left\{\left|\langle z, x_1^*\otimes \cdots \otimes \cdots x_k^*\rangle\right|:\, x_i^*\in\Ball(X_i^*),\, i\in [k]\right\}.
\end{align} We denote by $X_1\otimes_\epsilon\cdots \otimes_\epsilon X_k$ the Banach space defined by the completion of the space $X_1\otimes\cdots \otimes X_k$ equipped with the injective norm.

It is now easy to check that if $X$ and $Y$ are finite-dimensional Banach spaces, for every linear map $T:X\rightarrow Y$ we have 
\begin{align}\label{eq:lin-epsilon-norm}
\|T:X\rightarrow Y\|=\|\hat{T}\|_{X^*\otimes_\epsilon Y}.
\end{align}

The \emph{projective tensor norm} of $z\in X_1\otimes\cdots \otimes X_k$ is defined as
\begin{align}\label{Def projective norm}
\|z\|_{X_1\otimes_\pi\cdots \otimes_\pi X_k}=\inf\left\{\sum_{i=1}^N \|x_i^1\|\cdots \|x_i^k\|:z=\sum_{i=1}^N x_i^1\otimes \cdots \otimes x_i^k,\, N\in \N, \,  x_i^j\in X_j,\, j\in [k],\, i\in [N]\right\}.
\end{align}We denote by $X_1\otimes_\pi\cdots \otimes_\pi X_k$ the Banach space defined by the completion of the space $X_1\otimes\cdots \otimes X_k$ equipped with the projective norm.

Both the injective and the projective tensor norms satisfy the metric mapping property: Given Banach spaces $Y_j$ and linear maps $T_j:X_j \rightarrow Y_j$ for $j=1,\cdots, k$, it holds that
\begin{align}
\|T_1\otimes \cdots \otimes T_k:X_1\otimes_\alpha\cdots \otimes_\alpha X_k\rightarrow Y_1\otimes_\alpha\cdots \otimes_\alpha Y_k\|=\|T_1\|\cdots \|T_k\|,
\end{align}for $\alpha=\epsilon,\, \pi$. 

In addition, the injective and the projective norms are dual to each other. That is, for finite-dimensional Banach spaces $X_1, \cdots, X_k$, 
\begin{align}
(X_1\otimes_\pi\cdots \otimes_\pi X_k)^*=X_1^*\otimes_\epsilon\cdots \otimes_\epsilon X_k^*\hspace{0.3 cm}\text{isometrically}.
\end{align}

It follows from the previous equation that
\begin{align}\label{duality pi_epsilon}
|z^*(w)|\leq  \|z^*\|_{X_1^*\otimes_\epsilon\cdots \otimes_\epsilon X_k^*}\|w\|_{X_1\otimes_\pi\cdots \otimes_\pi X_k} 
\end{align}for every $z^*\in X_1^*\otimes\cdots \otimes X_k^*$ and $w\in X_1\otimes\cdots \otimes X_k$.

We also introduce some notions about operator spaces~\cite{EffrosRuanBook, PisierBook}. An operator space $X$ is a closed subspace of the space of all bounded operators on a complex Hilbert space $\h$. For any such subspace the operator norm on $\Bounded(\h)$ automatically induces a sequence of \emph{matrix norms} $\|\cdot\|_d$ on $M_d(X)$ via the inclusions  $M_d(X) \subseteq M_d(\Bounded(\h))\simeq \Bounded(\h^{\oplus d})$. 
In fact, according to Ruan's Theorem \cite{Ruan88}, an equivalent definition of an operator space can be given as a complex Banach space equipped with a sequence of matrix norms $(M_d(X),\|\cdot\|_d)$ satisfying certain conditions.

Given a linear map between operator spaces $T:X\rightarrow Y$, let $T_d$ denote the linear map
$$T_d:\,v=(v_{ij})\in M_d(X)\,\mapsto \,(\Id_{M_d} \otimes T)(v)=(T(v_{ij}))_{i,j}\in M_d(Y).$$
The map $T$ is said to be \emph{completely bounded} if its completely bounded norm if finite:
$$
\|T\|_{cb} \,:=\, \sup_{d\in\N} \|T_d\|\,<\, \infty.
$$
If $\|T\|_{cb}\leq 1$ the map is said to be a complete contraction.

Certain Banach spaces have a natural operator spaces structure (o.s.s.). This happens for the case of C$^*$-algebras which have a canonical inclusion in $\Bounded(\h)$ for a certain Hilbert space $\h$ obtained from the GNS construction (see for instance \cite[Theorem I.9.6]{Davison}). In particular, if we consider $\Bounded(\C^n)=S_\infty^n$, we obtain $M_d(S_\infty^n)=S_\infty^{dn}$ for every $d\geq 1$.

Given an operator space $X$ there is a natural o.s.s.\ on $X^*$, the dual space of $X$. The norms on $M_d(X^*)$ are specified through the identification (\ref{eq:lin-tensor}). This leads to the sequence of norms
\begin{align}\label{dual o.s.s.}
\|z\|_{M_d(X^*)}=\|T_z:X\rightarrow M_d\|_{cb},\qquad d\geq 1,
\end{align}which satisfies Ruan's theorem and, hence, defines an o.s.s. on $X^*$. In this way, we can endow $S_1(\h)$ with a natural o.s.s. as a dual space of $S_\infty(\h)$.

Given $k$ operator spaces $X_1,\cdots , X_k$, define the \emph{min norm} of $z\in X_1\otimes\cdots \otimes X_k$ as 
\begin{align}\label{Def min-norm}
\|z\|_{X\otimes_{min} \cdots \otimes_{min} X_k}=\sup \left\{\left\|(u_1\otimes\cdots \otimes u_k)(z)\right\|_{\B(\h_1\otimes \cdots\otimes \h_k)}:\, \|u_i:X_i\rightarrow \B(\h_i)\|_{cb}\leq 1, \, i\in [k]\right\},
\end{align}where the previous supremum runs over all possible complex Hilbert spaces $\h_1,\cdots, \h_k$. In fact, it can be seen that the supremum does not change if we restrict these spaces to be finite-dimensional. We denote by $X\otimes_{min} \cdots \otimes_{min} X_k$ the Banach space\footnote{Actually, this space has a natural o.s.s., but this will not be relevant in this work.} defined by the completion of the space $X_1\otimes\cdots \otimes X_k$ equipped with the min norm. 

The min norm satisfies the metric mapping property in the category of operator spaces. That is, given operator spaces $Y_j$ and linear maps $T_j:X_j \rightarrow Y_j$ for $j=1,\cdots, k$, it holds that
\begin{align}
\|T_1\otimes \cdots \otimes T_k:X_1\otimes_{min} \cdots \otimes_{min} X_k\rightarrow Y_1\otimes_{min} \cdots \otimes_{min} Y_k\|_{cb}=\|T_1\|_{cb}\cdots \|T_k\|_{cb}.
\end{align}

It can be seen that for every operator space $X$ and $z\in M_d(X)$, we have $$\|z\|_{M_d(X)}=\|z\|_{S_\infty^d\otimes_{min}X}.$$The min norm allows us to extend the identification (\ref{eq:lin-tensor}) to the category of operator spaces. Indeed, for every linear map between finite-dimensional operator spaces $T:X\rightarrow Y$ we have 
\begin{equation}\label{eq:lin-min-norm}
\|T:X\rightarrow Y\|_{cb}=\|\hat{T}\|_{X^*\otimes_{min} Y}.
\end{equation}

The particular case $X=S_1(\h)$, $Y=S_\infty(\h')$, where $\h$ and $\h'$ are finite-dimensional complex Hilbert spaces,  will be very relevant for us. Note that in this case $$\|T:S_1(\h)\rightarrow S_\infty(\h')\|_{cb}=\|\hat{T}\|_{S_\infty(\h)\otimes_{min} S_\infty(\h')}=\|\hat{T}\|_{S_\infty(\h\otimes \h')}.$$Here, according to Eq. (\ref{eq:lin-min-norm}) and the duality between $S_1(\h)$ and $S_\infty(\h)$, to a given a tensor $\hat{T}\in S_\infty(\h)\otimes_{min} S_\infty(\h')$ we associate the linear map $T:S_1(\h)\rightarrow S_\infty(\h')$ defined as
\begin{align}\label{Eq. Map-Tensor}
T(x)=(tr_{\mathcal B(\h)}\otimes \uno_{\B(\h')})\left(\hat{T}\cdot(x^T\otimes \uno_{\B(\h')})\right)\hspace{0.2 cm} \text{for every} \hspace{0.2 cm} x\in S_1(\h),
\end{align}where the product is defined as $(x\otimes y)\cdot (z\otimes w)=xz\otimes yw$. From now on, we just denote the \emph{partial trace} $(tr_{\mathcal B(\h)}\otimes \uno_{\B(\h')})=tr_{\mathcal B(\h)}$. Hence, the equation above is written as $T(x)=tr_{\mathcal B(\h)}\left(\hat{T}\cdot(x^T\otimes \uno_{\B(\h')})\right)$.

Since the space $S_1(\h_1)\otimes_{min}\cdots \otimes_{min}S_1(\h_k)$ will play an important role in this work, let us write the minimal norm of an element $x\in S_1(\h_1)\otimes\cdots \otimes S_1(\h_k)$ in detail:
\begin{lemma}\label{min on S_1}
Given $x\in S_1(\h_1)\otimes \cdots \otimes S_1(\h_k)$, we have 
\begin{align*}
\|x\|_{S_1(\h_1)\otimes_{min}\cdots \otimes_{min}S_1(\h_k)}=\sup\left \{\left|\left\langle\psi\left|tr_{\B(\h_1\otimes \cdots \otimes \h_k)}\left(U_1\otimes \cdots \otimes U_k\right)\left(x\otimes \uno_{\mathcal B(\h'_1\otimes\cdots \otimes \h'_k)}\right)\right|\eta\right\rangle\right|\right\},
\end{align*}where this supremum runs over all finite-dimensional complex Hilbert spaces $\h'_i$, all unitary matrices $U_i\in \mathcal B(\h_{i}\otimes \h'_i)$ for all $i=1,\cdots, k$ and all unit vectors $|\psi\rangle,\, |\eta \rangle\in \h'_1\otimes\cdots\otimes\h'_k$.
\end{lemma}
\begin{proof}
According to the definition of the min norm (\ref{Def min-norm}), we have 
\begin{align*}
\|x\|_{S_1(\h_1)\otimes_{min}\cdots \otimes_{min}S_1(\h_k)}=\sup\left \{\left\|(U_1\otimes \cdots \otimes U_k)(x)\right\|_{S_\infty(\h'_1\otimes\cdots \otimes \h'_k)}\right\},
\end{align*}where this supremum runs over all finite-dimensional complex Hilbert spaces $\h'_i$ and all completely contractive maps $u_i:S_1(\h_i)\rightarrow S_\infty(\h'_i)$, $i=1,\cdots, k$. Moreover, according to the Eq. (\ref{Eq. Map-Tensor}) we can write the previous expression as  
\begin{align*}
\sup\left \{\left\|tr_{\B(\h_1\otimes \cdots \otimes \h_k)}\left(\left(U_1\otimes \cdots \otimes U_k\right)\left(x^T\otimes  \uno_{\mathcal B(\h'_1\otimes\cdots \otimes \h'_k)}\right)\right)\right\|_{S_\infty(\h'_1\otimes\cdots \otimes \h'_k)}\right\},
\end{align*}where now $U_i\in \mathcal B(\h_{i}\otimes \h'_i)$ is a contraction for all $i=1,\cdots, k$. Now, according to Russo-Dye theorem (\cite[Theorem I.8.4]{Davison}), by convexity, we can assume that the $U_i$'s are unitary matrices in the previous supremum. Moreover, using that replacing $U_i$ by $U_i^T$ doesn't change this supremum, it is easy to see that the previous expression can be written as
\begin{align*}
\sup\left \{\left\|tr_{\B(\h_1\otimes \cdots \otimes \h_k)}\left(\left(U_1\otimes \cdots \otimes U_k\right)\left(x\otimes  \uno_{\mathcal B(\h'_1\otimes\cdots \otimes \h'_k)}\right)\right)\right\|_{S_\infty(\h'_1\otimes\cdots \otimes \h'_k)}\right\}.
\end{align*} Taking the supremum over unit vectors $|\psi\rangle,\, |\eta \rangle\in \h'_1\otimes\cdots\otimes \h'_k$ is now obvious from the definition of the operator norm.
\end{proof}

Finally, let us recall some basic results about Gaussian variables. Let $\{g_{i,j}\}_{i,j=1}^{N,M}$ be a family of independent normal real Gaussian variables. The following estimates are very well known:
\begin{small}\begin{align}
\mathbb E\left\|\sum_{i,j=1}^{N,M}g_{i,j}e_i\otimes e_j\right\|_{S_\infty^{N,M}}\leq C\min\left\{\sqrt{N}, \sqrt{M}\right\}, \hspace{0.3 cm} \mathbb E\left\|\sum_{i,j=1}^{N,M}g_{i,j}e_i\otimes e_j\right\|_{S_1^{N,M}}\leq C\sqrt{NM}\min\left\{\sqrt{N}, \sqrt{M}\right\},
\end{align}\end{small}where here $C$ is a universal constant.

We will also make use of Chevet's inequality \cite[Theorem~43.1]{Tomczak}.
\begin{theorem}[Chevet's inequality]\label{Chevet}
Let $X$ and $Y$ be Banach spaces. Define the Gaussian random tensor $z=\sum_{i=1}^m \sum_{j=1}^n g_{ij}x_i\otimes y_j\in X\otimes Y$, where $(g_{ij})_{i,j}$ are independent normal real Gaussian variables, and $(x_i)_{i=1}^m \subset X$, $(y_j)_{j=1}^n\subset Y$ are sequences of elements. Then,
\begin{align*}
\mathbb E \left\|z\right\|_{X \otimes_{\e} Y}\leq b\left(\sup_{x^*\in B_{X^*}}\left(\sum_{i=1}^m |x^*(x_i)|^2\right)^{\frac{1}{2}} \mathbb E\left\|\sum_{j=1}^n g_j y_j \right\|_{Y}+ \sup_{y^*\in B_{Y^*}}\left(\sum_{j=1}^n |y^*(y_j)|^2\right)^{\frac{1}{2}}  \mathbb E\left\|\sum_{i=1}^m g_i x_i \right\|_{X}\right),
\end{align*}where $b=1$ if the spaces $X$ and $Y$ are real and $b=4$ if they are complex and $(g_{i})_i$ is a sequence of independent normal real Gaussian random variables.
\end{theorem}

Note that, given a Banach space $Z$ and $(z_i)_{i=1}^n\subset Z$, we have $$\sup_{z^*\in B_{Z^*}}\left(\sum_{i=1}^n |z^*(z_i)|^2\right)^{\frac{1}{2}}=\|T:\ell_2^n\rightarrow Z\|,$$where $T$ is the linear map defined by $T(e_i) =z_i$ for every $i=1,\cdots, n$, and $(e_i)_i$ is any orthonormal basis of $\ell_2^n$.

\section{Quantum XOR games and strategies}\label{Sec: correlations}

\subsection{Quantum XOR games}

A $k$-player quantum XOR game is described by means of a family of $N$ different $k$-partite quantum states $(\rho_x)_{x=1}^N$, a family of signs $c=(c_x)_{x=1}^N\in \{-1,1\}^N$ and a probability distribution $p=(p_x)_x$ on $\{1,\cdots, N\}$. Here a $k$-partite quantum state $\rho$ is a trace-one positive semidefinite operator acting on the tensor product of $k$ finite-dimensional complex Hilbert spaces $\h_{A_1}\otimes \cdots \otimes \h_{A_k}$.

In order to understand the game, we can think of $k$ players (spatially separated) and a referee. The game starts with the referee choosing one of the states $\rho_x$ according to the probability distribution $p$. Then, the referee sends the corresponding part of the state $\rho_x$ to any of the players (this should be understood as a quantum question). After receiving the states, the players must answer an output, $a_i=\pm 1$ in the case of the player $A_i$. Then, the players win the game if $a_1\cdots a_k=c_x$. These games were first considered in \cite{ReVi15} for the case of two players as a natural generalization of classical XOR games, which have a great relevance in both quantum information and computer science. As we will see below, the relevant information of the game is encoded in the operator
\begin{align}\label{XOR quantum-operator}
G=\sum_{x=1}^Nc_xp_x \rho_x,
\end{align} which is a self-adjoint operator acting on $\h_{A_1}\otimes \cdots \otimes \h_{A_k}$ such that $\|G\|_{S_1(\h_{A_1}\otimes \cdots \otimes  \h_{A_k})}\leq 1$. That is, for any quantum XOR game the element $G$ defined in Eq. (\ref{XOR quantum-operator}) belongs to the unit ball of $S_1^{\mathrm{sa}}(\h_{A_1}\otimes \cdots \otimes  \h_{A_k})$. It is easy to see that the converse is also true: any self-adjoint operator $G$ such that $\|G\|_{S_1(\h_{A_1}\otimes \cdots \otimes  \h_{A_k})}\leq 1$ admits a decomposition as in Eq. (\ref{XOR quantum-operator}). Thus, the set of quantum XOR games can be identified with the unit ball of $S_1^{\mathrm{sa}}(\h_{A_1}\otimes \cdots \otimes  \h_{A_k})$.

For our purposes it suffices to notice that for any element $x\in S_1^{\mathrm{sa}}(\h_{A_1}\otimes \cdots \otimes  \h_{A_k})$, we can define a quantum XOR game by just considering: $$G=\frac{x}{\|x\|_{S_1(\h_{A_1}\otimes \cdots \otimes  \h_{A_k})}}.$$ 

In fact, it should be noted that the decomposition of a given $G$ as in Eq. (\ref{XOR quantum-operator}) is not in general unique. Therefore, an element $G$ in the unit ball of $S_1^{\mathrm{sa}}(\h_{A_1}\otimes \cdots \otimes  \h_{A_k})$ can describe different games. It is very easy to see that one such $G$ can always be understood as a quantum XOR game with only two quantum questions ($N=2$). Note that, since the Hilbert spaces $\h_{A_i}$'s are finite-dimensional, we have $S_1(\h_{A_1}\otimes \cdots \otimes  \h_{A_k})=S_1(\h_{A_1})\otimes \cdots \otimes  S_1(\h_{A_k})$.

A \emph{strategy} for $k$ players $A_1, \cdots , A_k$ is a linear map $\mathcal P:\mathcal B(\h_{A_1}\otimes \cdots \otimes  \h_{A_k})\rightarrow \R_+^{2^k}$ such that, for any given state $\rho$ acting on $\h_{A_1}\otimes \cdots \otimes  \h_{A_k}$, it assigns a probability distribution over the possible answers $$\mathcal P(\rho)=\mathcal P(a_1,\cdots ,a_k|\rho)_{a_1,\cdots ,a_k=\pm 1}.$$Note that, for a fixed strategy, the probability of winning the game is
\begin{align*}
\mathbf{P}_{win}(G;\mathcal P)&=\sum_{x:c_x=1}p_x\sum_{\substack{a_1,\cdots ,a_k:\\\substack{a_1\cdots a_k=1}}}P(a_1,\cdots ,a_k|\rho_x)\\&+\sum_{x:c_x=-1}p_x\sum_{\substack{a_1,\cdots ,a_k:\\\substack{a_1\cdots a_k=-1}}}P(a_1,\cdots ,a_k|\rho_x).
\end{align*}

It is easy to see that if all the players answer randomly (somehow the most naive strategy); that is $\mathcal P(a_1,\cdots ,a_k|\rho_x)=1/2^k$ for every $a_1,\cdots ,a_k=\pm 1$ and every $\rho_x$, then $\mathbf{P}_{win}(G;\mathcal P)=1/2$. Hence, when working with XOR games, it is very common to study the so called  \emph{bias} of the game, $\beta (G; \mathcal P)=\mathbf{P}_{win}(G;\mathcal P)-1/2$ or, equivalently (up to a multiplicative constant $2$),
\begin{align*}
&\mathbf{P}_{win}(G;\mathcal P)-\mathbf{P}_{lose}(G;\mathcal P)\\&=\sum_{x:c_x=1}p_x\sum_{\substack{a_1,\cdots ,a_k:\\\substack{a_1\cdots a_k=1}}}P(a_1,\cdots ,a_k|\rho_x)+\sum_{x:c_x=-1}p_x\sum_{\substack{a_1,\cdots ,a_k:\\\substack{a_1\cdots a_k=-1}}}P(a_1,\cdots ,a_k|\rho_x)\\&-
\sum_{x:c_x=1}p_x\sum_{\substack{a_1,\cdots ,a_k:\\\substack{a_1\cdots a_k=-1}}}P(a_1,\cdots ,a_k|\rho_x)-\sum_{x:c_x=-1}p_x\sum_{\substack{a_1,\cdots ,a_k:\\\substack{a_1\cdots a_k=1}}}P(a_1,\cdots ,a_k|\rho_x)\\&=\sum_{x=1}^Np_xc_x\sum_{a_1,\cdots ,a_k=\pm 1}a_1\cdots a_kP(a_1,\cdots ,a_k|\rho_x).
\end{align*}

Then, we see that in order to compute the bias of the game $G$ the only relevant part of the strategies are the correlations. That is, given a strategy $\mathcal P$ and a state $\rho$, if we define $$\gamma(\rho)=\sum_{a_1,\cdots ,a_k=\pm 1}a_1\cdots a_kP(a_1,\cdots ,a_k|\rho),$$we have
\begin{align}\label{bias}
\beta (G; \mathcal P)=\sum_{x=1}^Np_xc_x\gamma(\rho_x).
\end{align}

\subsection{Entangled strategies}

As the reader may guess, the winning probability of the game (and so its bias) will strongly depend on the type of the strategies, which are determined by the \emph{resources} allowed to the players to play the game. In this work, we are particularly interested in two type of strategies. The first one, usually called \emph{entangled} strategies, are those where the players are allowed to use a $k$-player quantum state. In this case, which was deeply studied in \cite{ReVi15} in the bipartite case, a strategy is defined by some POVMs $(E_1^{a_1})_{a_1=\pm 1},\cdots , (E_k^{a_k})_{a_k=\pm 1}$ acting on $\h_{A_1}\otimes \h_{A_1'},\cdots , \h_{A_k}\otimes \h_{A_k'}$ respectively and a $k$-player quantum states $\rho_{A_1'\cdots A_k'}$ acting on $\h_{A_1'}\otimes\cdots \otimes \h_{A_k'}$  such that 
\begin{align}\label{Entangled strategy}
P(a_1,\cdots ,a_k|\rho)=tr\left((E_1^{a_1}\otimes \cdots \otimes E_k^{a_k})(\rho\otimes \rho_{A_1'\cdots A_k'})\right).
\end{align}

Then, one can check that the entangled bias of the game $G$; that is, the largest bias of the game under these kinds of strategies, can be expressed as
\begin{align}\label{Entangled bias eq}
\beta^*(G)=\sup\left\{tr\left((A_1\otimes \cdots \otimes A_k)(G\otimes \rho_{A_1'\cdots A_k'})\right)\right\},
\end{align}where the supremum runs over all finite-dimensional Hilbert spaces $\h_{A_1'},\cdots ,\h_{A_k'}$, all $k$-player quantum states $\rho_{A_1'\cdots A_k'}$ acting on $\h_{A_1'}\otimes\cdots \otimes \h_{A_k'}$ and all self-adjoint operators $A_1,\cdots, A_k$ acting on $\h_{A_1}\otimes \h_{A_1'},\cdots , \h_{A_k}\otimes \h_{A_k'}$ respectively with norm lower than or equal to one.

\begin{lemma}\label{entangled bias}
Given a $k$-player quantum XOR game defined by the element $G\in S_1^{\mathrm{sa}}(\h_{A_1}\otimes \cdots \otimes  \h_{A_k})$, we have 
$$\beta^*(G)=\|G\|_{S_1(\h_{A_1})\otimes_{min}\cdots\otimes_{min}S_1(\h_{A_k})}.$$
\end{lemma}

The proof of this lemma is completely analogous to \cite[Claim 4.14]{ReVi15} in the bipartite case. We include the proof here for completeness.
\begin{proof}
According to Lemma \ref{min on S_1}, we have that 
\begin{align}
\|G\|_{S_1(\h_1)\otimes_{min}\cdots \otimes_{min}S_1(\h_k)}&=\sup\left \{\left|\left\langle\psi\left|tr_{\B(\h_1\otimes \cdots \otimes \h_k)}\left(\left(U_1\otimes \cdots \otimes U_k\right)\left(G\otimes \uno_{\mathcal B(\h'_1\otimes\cdots \otimes \h'_k)}\right)\right)\right|\eta\right\rangle\right|\right\}\\&=\sup\left \{\left| tr\left(\left(U_1\otimes \cdots \otimes U_k\right)\left(G\otimes |\eta\rangle\langle\psi|\right)\right)\right|\right\},
\end{align}where this supremum runs over all finite-dimensional complex Hilbert spaces $\h'_i$ and all unitary matrices $U_i\in \mathcal B(\h_{i}\otimes \h'_i)$ for all $i=1,\cdots, k$ and all unit vectors $|\psi\rangle,\, |\eta \rangle\in \h'_1\otimes\cdots\otimes \h'_k$. Hence, according to Eq. (\ref{Entangled bias eq}),  it is clear that $$\beta^*(G)\leq \|G\|_{S_1(\h_{A_1})\otimes_{min}\cdots\otimes_{min}S_1(\h_{A_k})}.$$

In order to prove the converse inequality, given some Hilbert spaces $\h'_i$, some unitary matrices $U_i\in \mathcal B(\h_{i}\otimes \h'_i)$ for all $i=1,\cdots, k$ and unit vectors $|\psi\rangle,\, |\eta \rangle\in \h'_1\otimes\cdots\otimes \h'_k$, first note that we can assume that $$\alpha:=\left\langle\psi\left|tr_{\B(\h_1\otimes \cdots \otimes \h_k)}\left(\left(U_1\otimes \cdots \otimes U_k\right)\left(G\otimes \uno_{\mathcal B(\h'_1\otimes\cdots \otimes \h'_k)}\right)\right)\right|\eta\right\rangle
$$ is a non-negative real number. Indeed, if this is not the case we just define $|\tilde{\eta}\rangle=\frac{\alpha}{|\alpha|}|\eta\rangle$, which is again a unit vector, and obtain that the corresponding expression equals $|\alpha|$.

Then, we can define 
$$\tilde{U}_i=\left(
\begin{array}{cc}
0&U_i\\
 U_i^*&0
\end{array}
\right)=e_{12}\otimes U_i+e_{21}\otimes U_i^*\in M_2(\mathcal B(\h_{i}\otimes \h'_i))=\mathcal B(\C^2\otimes \h_{i}\otimes \h'_i)$$and $$|\tilde{\psi}\rangle=\frac{1}{\sqrt{2}}\left(|\psi\rangle\otimes |1\cdots 1\rangle+|\eta \rangle\otimes |2\cdots 2\rangle\right)\in \h'_1\otimes\cdots\otimes \h'_k\otimes \C^{2^k},$$where we have denoted $|1\cdots 1\rangle=e_1\otimes \cdots \otimes e_1$ and $|2\cdots 2\rangle=e_2\otimes \cdots \otimes e_2$.

Then, by denoting $\tilde{\h}_i=\C^2\otimes \h'_i$ for every $i=1,\cdots, k$, we obtain
\begin{align*}
\beta^*(G)&\geq tr\left((\tilde{U}_1\otimes \cdots \otimes \tilde{U}_k)\left(G\otimes |\tilde{\psi}\rangle\langle\tilde{\psi}|\right)\right)\\&=\left\langle \tilde{\psi}\left|tr_{\B(\h_1\otimes \cdots \otimes \h_k)}\left(\left(\tilde{U}_1\otimes \cdots \otimes\tilde{U}_k\right)\left(G\otimes \uno_{\mathcal B(\tilde{\h}_1\otimes\cdots \otimes \tilde{\h}_k)}\right)\right)\right|\tilde{\psi}\right\rangle\\&=\frac{1}{2}\Big[\left\langle \psi\left|tr_{\B(\h_1\otimes \cdots \otimes \h_k)}\left(\left(U_1\otimes \cdots \otimes U_k\right)\left(G\otimes \uno_{\mathcal B(\h'_1\otimes\cdots \otimes \h'_k)}\right)\right)\right|\eta\right\rangle\\&+\left\langle \eta\left|tr_{\B(\h_1\otimes \cdots \otimes \h_k)}\left(\left(U_1^*\otimes \cdots \otimes U_k^*\right)\left(G\otimes \uno_{\mathcal B(\h'_1\otimes\cdots \otimes \h'_k)}\right)\right)\right|\psi\right\rangle\Big].
\end{align*}

Now, using that $\alpha\in \R$ and that $G^*=G$ one can check that the last expression equals $$\left\langle\psi\left|tr_{\B(\h_1\otimes \cdots \otimes \h_k)}\left(\left(U_1\otimes \cdots \otimes U_k\right)\left(G\otimes \uno_{\mathcal B(\h'_1\otimes\cdots \otimes \h'_k)}\right)\right)\right|\eta\right\rangle.$$

Since this happens for all Hilbert spaces $\h_i$, all unitaries $U_i\in \mathcal B(\h_{i}\otimes \h'_i)$,  $i=1,\cdots, k$ and all unit vectors $|\psi\rangle,\, |\eta \rangle\in \h'_1\otimes\cdots\otimes \h'_k$, we conclude that $$\beta^*(G)\geq \|G\|_{S_1(\h_{A_1})\otimes_{min}\cdots\otimes_{min}S_1(\h_{A_k})}.$$
\end{proof}

\subsection{LOCC strategies and separable strategies}

Since we want to quantify the resource of classical communication in the context of quantum XOR games, we must consider those strategies which can be implemented by the players when they are allowed to perform local operations and also the interchange of classical communication. This leads us to the following definition.
\begin{definition}
Given a POVM $(E_{a_1,\cdots, a_k})_{a_1,\cdots, a_k=\pm 1}$, we say that it is a \emph{LOCC measurement} if it can be implemented by means of local operations and the interchange of classical communication between the $k$ parties. We denote by LOCC the set of LOCC measurement\footnote{Note that we are implicitly fixing the number of possible outputs here $a_1,\cdots, a_k=\pm 1$. However, this is the only case we will consider in this paper. So, we prefer not to add more notation to make this point explicit.}.

We say that a given strategy $\mathcal P$ for a $k$-player quantum XOR game is a \emph{LOCC strategy} if it is of the form 
$$\mathcal P(\rho)=\left(tr\left(E_{a_1,\cdots, a_k}\rho\right)\right)_{a_1,\cdots ,a_k=\pm 1},$$where $(E_{a_1,\cdots, a_k})_{a_1,\cdots, a_k=\pm 1}$ is a LOCC measurement.
\end{definition}

Although one can write the form of a LOCC measurement explicitly \cite{CLMOW14}, the set LOCC is in general quite tricky to work with. We will introduce another type of measurements which will be very useful for us.

\begin{definition}\label{defi SEP}
Given a POVM $(E_{a_1,\cdots, a_k})_{a_1,\cdots, a_k=\pm 1}$, we say that it is a \emph{separable measurement} if it is of the form $$\left (\sum_{i\in I}P_i^{1,a_1}\otimes \cdots \otimes P_i^{k,a_k}\right)_{a_1,\cdots ,a_k=\pm 1},$$where $I$ is a finite set, $P_i^{j,a_j}\in \mathcal B(\h_{A_j})$ is a positive semidefinite operator for every $i$, $j$ and $a_{j}=\pm 1$ and such that $$\sum_{i\in I}\sum_{a_1,\cdots ,a_k=\pm 1}P_i^{1,a_1}\otimes \cdots \otimes P_i^{k,a_k}=\uno_{\mathcal B(\h_{A_1})}\otimes \cdots \otimes \uno_{\mathcal B(\h_{A_k})}.$$We denote by SEP the set of all separable measurements.

We say that a given strategy $\mathcal P$ for a $k$-player quantum XOR game is \emph{separable} if it is of the form 
$$\mathcal P(\rho)=\left(tr\left(E_{a_1,\cdots, a_k}\rho\right)\right)_{a_1,\cdots ,a_k=\pm 1},$$where $(E_{a_1,\cdots, a_k})_{a_1,\cdots, a_k=\pm 1}$ is a separable measurement.
\end{definition}

It is well known that $LOCC\subseteq SEP$ (see \cite[Section 2.3]{CLMOW14} for an explanation in the more general context of quantum instruments and the comments right below  \cite[Definition 12]{LPW18} for an explanation in the context of measurements for general GPTs). At the same time, the set SEP admits a simpler mathematical description than the set LOCC, and this can be used to derive some limitations on LOCC, as it has been shown in different contexts such as entanglement distillation \cite{HHH98, Rains01} and state discrimination \cite{Che04, DLT02}. We will use the set SEP in this work in the same way.

Give a $k$-player quantum XOR $G$, we denote by $\beta_{LOCC}(G)$ and $\beta_{SEP}(G)$ the maximum bias of $G$ attainable by the players when they are restricted to the use of LOCC strategies and separable strategies respectively. The following corollary will be key later.
\begin{corollary}\label{upper bound LOCC bias}
Given a $k$-player quantum XOR game defined by the element $G$, we have that 
$$\beta_{LOCC}(G)\leq \beta_{SEP}(G)\leq \sup\left\{tr\left( \sum_{i\in I}A_i^1\otimes \cdots \otimes A_i^kG\right)\right\},$$
where the sup is taken over all finite sets $I$ and all self-adjoint operators $A_i^j\in \mathcal B(\h_{A_j})$ for every $i$ and $j$ such that $$\sum_{i\in I}\|A_i^1\|_{S_1(\h_{A_1})}\cdots \|A_i^k\|_{S_1(\h_{A_k})}\leq n_1\cdots n_k.$$Here, $n_i=\text{dim}(\h_{A_i})$, $i=1,\cdots, k$.
\end{corollary}
\begin{proof}
The first inequality is a direct consequence of the already mentioned inclusion $LOCC\subseteq SEP$. For the second inequality, according to Eq. (\ref{bias})  and Definition \ref{defi SEP}, the bias of $G$ corresponding to a given separable strategy $\mathcal P$ is of the form
\begin{align*}
\beta (G; \mathcal P)&=\sum_{a_1,\cdots ,a_k=\pm 1}a_1\cdots a_k \, tr\left (\sum_{i\in I}P_i^{1,a_1}\otimes \cdots \otimes P_i^{k,a_k}G\right)\\&=tr\left (\sum_{i\in I}A_i^1\otimes \cdots \otimes A_i^kG\right),
\end{align*}where we have denoted $A_i^j=\sum_{a_j=\pm 1}a_jP_i^{j,a_j}\in \mathcal B(\h_{A_j})^{\mathrm{sa}}$ for every $i$ and $j$.

On the other hand, note that 
\begin{align*}
\sum_{i\in I}\|A_i^1\|_{S_1(\h_{A_1})}\cdots \|A_i^k\|_{S_1(\h_{A_k})}&= \sum_{i\in I}\left\|\sum_{a_1=\pm 1}a_1P_i^{1,a_1}\right\|_{S_1(\h_{A_1})}\cdots \left\|\sum_{a_k=\pm 1}a_kP_i^{k,a_k}\right\|_{S_1(\h_{A_k})}\\&\leq  \sum_{i\in I}\left(\sum_{a_1=\pm 1}\left\|P_i^{1,a_1}\right\|_{S_1(\h_{A_1})}\right)\cdots \left(\sum_{a_k=\pm 1}\left\|P_i^{k,a_k}\right\|_{S_1(\h_{A_k})}\right)\\&=\sum_{i\in I}\left(\sum_{a_1=\pm 1}tr(P_i^{1,a_1})\right)\cdots \left(\sum_{a_k=\pm 1}\tr(P_i^{k,a_k})\right)\\&=tr\left(\sum_{i\in I}\sum_{a_1,\cdots ,a_k=\pm 1}P_i^{1,a_1}\otimes \cdots \otimes P_i^{k,a_k}\right)\\&=tr\left(\uno_{\mathcal B(\h_{A_1})}\otimes \cdots \otimes \uno_{\mathcal B(\h_{A_k})}\right)\\&=n_1\cdots n_k.
\end{align*}
\end{proof}

\section{Proof of the main result}\label{Sec: Proofs}

From now on, we will fix $\h_1=\cdots=\h_k=\C^D$, so that $\mathcal B(\h_i)=S_\infty^D$ and $S_1(\h_i)=S_1^D$ for every $i$.

The key result of the paper is as follows: 
\begin{theorem}\label{Key Thm}
For every $k\geq 3$ there exists a constant $C(k)$ depending only on $k$ with the following property: Given natural numbers $D$ and $m$ satisfying that $D\geq c m^4 \sqrt{\log m}$ (where $c$ is a universal constant) we have  
\begin{align*}
\Big\|\frac{Id}{D^k}:\overbrace{S_\infty^{D}\otimes_\epsilon \cdots \otimes_\epsilon S_\infty^{D}}^k\rightarrow \overbrace{S_1^{D}\otimes_{min}\cdots \otimes_{min}S_1^{D}}^k\Big\|\geq C(k)m^{\frac{k}{2}-1}(\log m)^{-\frac{k}{2}}. 
\end{align*} 
\end{theorem}

\begin{remark}\label{explanation NCGT}
In contrast to the previous theorem, one can show that there exists a universal constant $C$ such that 
\begin{align}\label{bilinear case}
\Big\|\frac{Id}{D^2}:S_\infty^{D}\otimes_\epsilon S_\infty^{D}\rightarrow S_1^{D}\otimes_{min}S_1^{D}\Big\|\leq C. 
\end{align}This explains why Theorem \ref{Main Thm} cannot be extended automatically to the case of bipartite quantum XOR games. 

In order to show Eq. (\ref{bilinear case}), let us first note that, according to  Eq. (\ref{eq:lin-epsilon-norm}) and Eq. (\ref{eq:lin-min-norm}), Eq. (\ref{bilinear case}) is equivalent to state that for every linear map $T:S_1^D\rightarrow S_\infty^D$ we have $$\|T:S_\infty^D\rightarrow S_1^D\|_{cb}\leq CD^2\|T:S_1^D\rightarrow S_\infty^D\|.$$ 

In order to see this inequality, let us invoke \cite[Remark 3.1]{JKPV22} for the particular case $H=\ell_2^D$ and $a=b=Id_{H}$, to deduce that there exist linear maps $u:S_\infty^{D}\rightarrow R\cap C$ and $v:R\cap C\rightarrow S_1^{D}$ such that the identity map $Id:S_\infty^{D}\rightarrow S_1^{D}$ can be written as $Id=v\circ u$, with$$\|u:S_\infty^{D}\rightarrow R\cap C\|_{cb}\|v:\ell_2\rightarrow S_1^{D}\|\leq \sqrt{2}\|a\|_{S_2(H)}\|b\|_{S_2(H)}=\sqrt{2}D.$$Here $R\cap C$ denotes a particular operator spaces structure on the Hilbert space $\ell_2$. Then, following the main ideas in \cite{JKPV22}, we use a noncommutative Grothendieck's theorem (see Corollary 2.2 in its dual form) to state that, given any linear map $T:S_1^D\rightarrow S_\infty^D$, we have 
\begin{align*}
\|T:S_\infty^D\rightarrow S_1^D\|_{cb}&\leq \|u:S_\infty^D\rightarrow R\cap C\|_{cb}\| \|T\circ v:R\cap C\rightarrow S_1^D\|_{cb}\\&\leq 2\|u:S_\infty^D\rightarrow R\cap C\|_{cb}\| \|T\circ v:\ell_2\rightarrow S_1^D\|\\&\leq 2\|u:S_\infty^D\rightarrow R\cap C\|_{cb}\| \|v:\ell_2\rightarrow S_1^D\|\|T:S_1^D\rightarrow S_1^D\|\\&\leq 2\sqrt{2}D\|T:S_1^D\rightarrow S_1^D\|\\&\leq 2\sqrt{2}D^2\|T:S_1^D\rightarrow S_\infty^D\|.
\end{align*}

Hence, we see that Eq. (\ref{bilinear case}) is a direct consequence of a noncommutative Grothendieck's theorem for bilinear forms on C$^*$-algebras and Theorem \ref{Key Thm} can be understood as a counterexample of that theorem for $k$-linear forms if $k\geq 3$.
\end{remark}

Before proving Theorem \ref{Key Thm} let us explain why it implies Theorem \ref{Main Thm}.

\begin{proof}[Proof of Theorem \ref{Main Thm}]
According to Theorem \ref{Key Thm} there exists an element $z\in S_\infty^{D}\otimes \cdots\otimes S_\infty^{D}$ such that $$\|z\|_{S_\infty^{D}\otimes_\epsilon \cdots\otimes_\epsilon S_\infty^{D}}\leq 1\hspace{0.3 cm} \text{and} \hspace{0.3 cm} \left\|z\right\|_{S_1^{D}\otimes_{min}\cdots\otimes_{min}S_1^{D}}\geq C(k)D^km^{\frac{k}{2}-1}(\log m)^{-\frac{k}{2}},$$ where $m$ and $D$ satisfy the relation $D\geq c m^4 \sqrt{\log m}$ for a universal constant $c$.

Let us first prove that we can assume the element $z$ to be self-adjoint ($z\in M_{D^k}^{\mathrm{sa}}$) at the prize to replace $D$ by $2D$, which will be reflected in a slight modification of the constant $C(k)$. To this end consider the linear isometry $j:S_\infty^D\rightarrow S_\infty^{2D}$ given by $$j(x)=\left(
\begin{array}{cc}
0&x\\
 x^*&0
\end{array}
\right)$$and define the new element $\tilde{z}=(j\otimes \cdots\otimes j)(z)\in S_\infty^{2D,\mathrm{sa}}\otimes \cdots \otimes S_\infty^{2D, \mathrm{sa}}=S_\infty^{(2D)^k, \mathrm{sa}}$. 

Now, the metric mapping property of the $\epsilon$ norm guarantees that $$\|\tilde{z}\|_{S_\infty^{2D}\otimes_\epsilon \cdots\otimes_\epsilon S_\infty^{2D}}= \|z\|_{S_\infty^{D}\otimes_\epsilon \cdots\otimes_\epsilon S_\infty^{D}},$$while one can check that
\begin{align}\label{comparison norms}
\|\tilde{z}\|_{S_1^{2D}\otimes_{min} \cdots\otimes_{min}  S_1^{2D}}\geq \|z\|_{S_1^{D}\otimes_{min}\cdots\otimes_{min}S_1^{D}}.
\end{align}

For this last inequality, recall from Lemma \ref{min on S_1} that $$\|z\|_{S_1^{D}\otimes_{min}\cdots\otimes_{min}S_1^{D}}=\sup\left\{\left\|tr_{M_{D^k}}\left(\left(U_1\otimes \cdots\otimes U_k\right)\left(z^T\otimes \uno_{M_{d^k}}\right)\right)\right\|_{S_\infty^{d^k}}:\, d\in \N,\,  \sup_{i=1,\cdots , k}\|U_i\|_{S_\infty^{Dd}}\leq 1\right\}.$$

Now, if we consider the linear map $h:S_\infty^D\rightarrow S_\infty^{2D}$ given by $$h(x)=\left(
\begin{array}{cc}
0&x\\
0&0
\end{array}
\right),$$for any choice of matrices $U_i$ as above, we clearly have that $\tilde{U}_i=(h\otimes Id_{M_d})(U_i)$ satisfies  $\|\tilde{U}_i\|_{_{S_\infty^{(2D)d}}}\leq 1$ for every $i=1,\cdots , k$ and 
\begin{align*}
\|\tilde{z}\|_{S_1^{2D}\otimes_{min}\cdots\otimes_{min}S_1^{2D}}&\geq \left\|tr_{M_{(2D)^k}}\left(\left(\tilde{U}_1\otimes \cdots\otimes \tilde{U}_k\right)\left(\tilde{z}^T\otimes \uno_{M_{d^k}}\right)\right)\right\|_{S_\infty^{d^k}}\\&=\left\|tr_{M_{D^k}}\left(\left(U_1\otimes \cdots\otimes U_k\right)\left(z^T\otimes \uno_{M_{d^k}}\right)\right)\right\|_{S_\infty^{d^k}}.
\end{align*} Hence, we obtain inequality (\ref{comparison norms}).

Once we have shown that $z$ can be assumed to be self-adjoint, we know that $$G=\frac{z}{\|z\|_{S_1^{D^k}}}$$ defines a $k$-player quantum XOR game. Moreover, according to Lemma \ref{entangled bias} we know that $$\beta^*(G)= \frac{1}{\|z\|_{S_1^{D^k}}}\|z\|_{S_1^{D}\otimes_{min}\cdots\otimes_{min}S_1^{D}}\geq \frac{1}{\|z\|_{S_1^{D^k,\mathrm{sa}}}}C(k)D^km^{\frac{k}{2}-1}(\log m)^{-\frac{k}{2}}.$$

Let us finally show that $\beta_{SEP}(G)\leq \frac{D^k}{\|z\|_{S_1^{D^k}}}$, from where the statement of the theorem follows. According to Corollary \ref{upper bound LOCC bias}, we have 
\begin{align*}
\beta_{SEP}(G)\leq \sup\left\{tr\left( \sum_{i\in I}A_i^1\otimes \cdots \otimes A_i^kG\right)\right\},
\end{align*}where the sup is taken over all finite sets $I$ and all self-adjoint operators $A_i^j\in S_\infty^D$ for every $i$ and $j$ such that $$\sum_{i\in I}\|A_i^1\|_{S_1^{D}}\cdots \|A_i^k\|_{S_1^{D}}\leq D^k.$$

In our case, given a finite set $I$ and a family of self-adjoint operators $A_i^j\in S_\infty^D$ satisfying the previous requirements, we have
\begin{align}\label{last upper bound}
tr\left( \sum_{i\in I}A_i^1\otimes \cdots\otimes A_i^k \, G\right)&= \frac{1}{\|z\|_{S_1^{D^k}}}tr\left( \sum_{i\in I}A_i^1\otimes \cdots\otimes A_i^k \, z\right)\\\nonumber&\leq \frac{1}{\|z\|_{S_1^{D^k}}}\left\| \sum_{i\in I}A_i^1\otimes\cdots \otimes A_i^k\right\|_{S_1^D\otimes_{\pi}\cdots\otimes_{\pi}S_1^D}\|z\|_{S_\infty^{D}\otimes_\epsilon \cdots\otimes_\epsilon S_\infty^{D}}\\\nonumber&\leq \frac{1}{\|z\|_{S_1^{D^k}}}\sum_{i\in I}\|A_i^1\|_{S_1^{D}}\cdots \|A_i^k\|_{S_1^{D}}\\\nonumber&\leq \frac{D^k}{\|z\|_{S_1^{D^k}}},
\end{align}where the first equality follows from the definition of $G$, in the second inequality we have used eq. (\ref{duality pi_epsilon}) and the third inequality follows from the fact that $\|z\|_{S_\infty^{D}\otimes_\epsilon \cdots\otimes_\epsilon S_\infty^{D}}\leq 1$ and the triangle inequality.

Since Eq. (\ref{last upper bound}) works for any family of self-adjoint operators $A_i^j\in S_\infty^D$ satisfying the requirements in Corollary \ref{upper bound LOCC bias}, we conclude the proof.
\end{proof}

In order to prove Theorem \ref{Key Thm}, the following theorem, proved in \cite{BrVi13} and slightly improved in \cite{Pisier_Bell}, will be crucial for us (see \cite[Section 2]{Pisier_Bell} for details).
\begin{theorem}\label{Thm. Briet-Vidick}
Let $k$ be a natural number and $\{g_{i_1,\cdots, i_k}\}_{i_1,\cdots, i_k=1}^n$ be a family of independent normal real Gaussian variables. Let $\{g'_{i_1,\cdots, i_k}\}_{i_1,\cdots, i_k=1}^n$ be an independent copy of the family $\{g_{i_1,\cdots, i_k}\}_{i_1,\cdots, i_k=1}^n$. Then, if we denote $$\tau=\sum_{\substack{i_1,\cdots, i_k=1\\\substack{i'_1,\cdots, i'_k=1}}}^ng_{i_1,\cdots, i_k}g'_{i'_1,\cdots, i'_k}e_{i_1i_1'}\otimes\cdots\otimes e_{i_ki_k'},$$we have 
\begin{align*}
\mathbb E\|\tau\|_{\ell_2^{n^2}\otimes_\epsilon\cdots \otimes_\epsilon \ell_2^{n^2}}\leq C(k)n(\log n)^{\frac{k}{2}},
\end{align*}where $C(k)$ is a constant depending only on $k$. 
\end{theorem} 

Note that, given the random element $\tau$ from the previous theorem, one has in addition that
\begin{align*}
\mathbb E\|\tau\|_{S_\infty^{n^k}}&=\mathbb E\left\|\sum_{\substack{i_1,\cdots, i_k=1\\\substack{i'_1,\cdots, i'_k=1}}}^ng_{i_1,\cdots, i_k}g'_{i'_1,\cdots, i'_k}(e_{i_1}\otimes\cdots\otimes e_{i_k})\otimes (e_{i_1'}\otimes\cdots\otimes e_{i_k'})\right\|_{S_\infty^{n^k}}\\&=\left(\mathbb E\left\|\sum_{i_1,\cdots, i_k=1}^ng_{i_1,\cdots, i_k}e_{i_1}\otimes \cdots \otimes e_{i_k}\right\|_{\ell_2^{n^k}}\right)^2\geq C'n^k,
\end{align*}where $C'$ is a universal constant.

Hence, we can easily deduce the existence of an element (that we will also denote $\tau$) such that 
\begin{align}\label{element Briet-Vidick}
\|\tau\|_{\ell_2^{n^2}\otimes_\epsilon\cdots \otimes_\epsilon \ell_2^{n^2}}\leq C(k)n(\log n)^{\frac{k}{2}} \hspace{0.3 cm}\text{and}  \hspace{0.3 cm} \|\tau\|_{S_\infty^{n^k}}\geq C'n^k.
\end{align}

In fact, one can show that Eq. (\ref{element Briet-Vidick}) happens with high probability.

We will also need the following lemma about Gaussian variables.
\begin{prop}\label{prop_gaussians}
Let $(g_j^l)_{j,l=1}^{m,N}$ be a family of independent normal real Gaussian variables. Then, for every $j,j'=1,\cdots, m$ and every $t\geq 0$ we have 
\begin{align*}
\mathbb P\left(\left|\frac{1}{N}\sum_{l=1}^Ng_j^lg_{j'}^l-\delta_{jj'}\right|\geq t\right) \leq  2e^{-cN\min\{\frac{t^2}{C^2}, \frac{t}{C}\}}, 
\end{align*}where $c$ and $C$ are universal constants.

Hence,
\begin{align*}
\mathbb P\left(\max_{j,j'=1,\cdots, m}\left|\frac{1}{N}\sum_{l=1}^Ng_j^lg_{j'}^l-\delta_{jj'}\right|\geq t\right) \leq  2e^{2\log m-cN\min\{\frac{t^2}{C^2}, \frac{t}{C}\}}.
\end{align*}

\end{prop}
\begin{proof}
It is well known that the product of two Gaussian variables $h_l=g_j^lg_{j'}^l$ is a sub-exponential random variable (\cite[Lemma 2.7.7]{Vershynin_book}). That is, there exists a constant $K$ such that for every $t\geq 0$ we have $$\mathbb P\left(|h_l|\geq t\right)\leq 2e^{-\frac{t}{K}}.$$

Moreover, the centered random variable $h_l=g_j^lg_{j'}^l-\mathbb E g_j^lg_{j'}^l=g_j^lg_{j'}^l-\delta_{jj'}$ is also sub-exponential (see \cite[Lemma 2.7.10]{Vershynin_book}). Hence, according to Bernstein's inequality \cite[Corollary 2.8.3]{Vershynin_book}, we deduce that for every $j$, $j'$ and $t\geq 0$ we have
\begin{align*}
\mathbb P\left(\left|\frac{1}{N}\sum_{l=1}^Ng_j^lg_{j'}^l-\delta_{jj'}\right|\geq t\right)=\mathbb P\left(\left|\frac{1}{N}\sum_{l=1}^Nh_l\right|\geq t\right) \leq  2e^{-cN\min\{\frac{t^2}{C^2}, \frac{t}{C}\}}, 
\end{align*}where $c$ and $C$ are universal constants.

The last statement of the proposition follows easily from a union bound.
\end{proof}

Finally, the following result will be also very useful.
\begin{lemma}\label{Auerbach_operator-spaces}
Let $\{b_j\}_{j=1}^d$, $\{b_j^*\}_{j=1}^d$ be an Auerbach basis of a $d$-dimensional operator space $X$ and let $\{c_j\}_{j=1}^d$ be some elements in $X$ such that $$\sup_j\|b_j-c_j\|\leq \frac{\epsilon}{d},$$where $\epsilon \in [0,1)$.

Then, there exists a linear isomorphism $\psi:X\rightarrow X$ satisfying $\psi(b_j)=c_j$ for every $j$ and such that $$\max\{\|\psi\|_{cb}, \|\psi^{-1}\|_{cb}\}\leq \frac{1}{1-\epsilon}.$$
\end{lemma}
\begin{proof}
The case $\epsilon=0$ is trivial, so we assume $\epsilon\in (0,1)$. 
Let us define the linear map $R:X\rightarrow X$ associated to the tensor $R=\sum_{j=1}^db_j^*\otimes (b_j-c_j)$. That is, $R(x)=\sum_{j=1}^db_j^*(x)(b_j-c_j)$ for every $x\in X$. It is clear that $$\|R\|_{cb}\leq \sum_{j=1}^d\|b_j^*\|_{X^*}\|b_j-c_j\|_X\leq \epsilon.$$

Then, we define the linear map $\psi=Id-R=\sum_{j=1}^db_j^*\otimes c_j$, which satisfies that $\psi(b_j)=c_j$ for every $j$ and $$\|\psi\|_{cb}\leq \|Id\|_{cb} +\|R\|_{cb}\leq 1+ \epsilon.$$

Finally, the cb norm of $\psi^{-1}$ can be upper bounded by using the Neumann series:
$$\|\psi^{-1}\|_{cb}=\|(Id-R)^{-1}\|_{cb}\leq \sum_{k=0}^\infty \|R\|^{k}_{cb}\leq \sum_{k=0}^\infty \epsilon^{k}=\frac{1}{1-\epsilon}.$$

Hence, $$\max\{\|\psi\|_{cb}, \|\psi^{-1}\|_{cb}\}\leq \max\left\{1+ \epsilon,\frac{1}{1-\epsilon}\right\}\leq \frac{1}{1-\epsilon}.$$
\end{proof}

We are now ready to prove our main theorem.
\begin{proof}[Proof of Theorem \ref{Key Thm}]
Let $D$ and $m$ be natural numbers such that $D\geq m^2$ and let us define, for every $t,v=1,\cdots, m$, the random element $$f_{tv}=\frac{1}{D^\frac{3}{2}}\sum_{r,s=1}^Dg_{rs}^{tv}e_{rs}\in S_\infty^D,$$where $(g_{rs}^{tv})_{r,s;t,v=1}^{D,m}$ are independent normal real Gaussian variables. If we consider the linear map $u:\ell_2^{m^2}\rightarrow S_\infty^D$, defined by $u(e_{tv})=f_{tv}$ for ever $t,v$, Chevet's inequality (Theorem \ref{Chevet}) allows us to upper bound the expected value
\begin{align}\label{Markov1}
\mathbb E\|u:\ell_2^{m^2}\rightarrow S_\infty^D\|=\frac{1}{D^\frac{3}{2}}\mathbb E \left\|\sum_{t,v=1}^m\sum_{r,s=1}^{D}g_{rs}^{tv}e_{tv}\otimes e_{rs}\right\|_{\ell_2^{m^2}\otimes_{\epsilon}S_\infty^D}\leq \frac{C_1}{D},
\end{align}where $C_1$ is a universal constant.

On the other hand, according to the Eqs. (\ref{eq:lin-min-norm}) and (\ref{min on S_1}), the map $\Phi:S_1^D\rightarrow S_\infty^m$, defined by $$\Phi(x)=\frac{1}{\sqrt{D}}\sum_{r,s=1}^D\sum_{t,v=1}^mg_{rs}^{tv}tr(x^Te_{rs})e_{tv},$$ satisfies
\begin{align}\label{Markov2}
\mathbb E\|\Phi:S_1^D\rightarrow S_\infty^m\|_{cb}=\mathbb E\left\|\frac{1}{\sqrt{D}}\sum_{r,s=1}^D\sum_{t,v=1}^mg_{rs}^{tv}e_{rs}\otimes e_{tv}\right\|_{S_\infty^{Dm}}\leq C_2\sqrt{m},
\end{align}where $C_2$ is a universal constant. Here, we have used that $S_\infty^D\otimes_{min} S_\infty^m=S_\infty^{Dm}$.

According to Markov's inequality, by slightly modifying the constant $C_1$ and $C_2$ above we can assume that the random variables $(g_{rs}^{tv})_{r,s,t,v}$ satisfy
\begin{align}\label{Markov1and2}
\|u:\ell_2^{m^2}\rightarrow S_\infty^D\|\leq \frac{C_1}{D}\hspace{0.3 cm} \text{and}\hspace{0.3 cm}\|\Phi:S_1^D\rightarrow S_\infty^m\|_{cb}\leq C_2\sqrt{m}
\end{align} with probability larger than 1/2.

On the other hand, $$\Phi(f_{tv})=\frac{1}{D^2}\sum_{t',v'=1}^m\left(\sum_{r,s=1}^Dg_{rs}^{tv}g_{rs}^{t'v'}\right)e_{t'v'}.$$

Now, if we consider the Auerbach basis $(e_{tv})_{t,v=1}^m$ of $S_\infty^m$,  we note that 
\begin{align*}
\|\Phi(f_{tv})-e_{tv}\|_{S_\infty^m}&=\left\|\sum_{t',v'=1}^m\left[\left(\frac{1}{D^2}\sum_{r,s=1}^Dg_{rs}^{tv}g_{rs}^{t'v'}\right)-\delta_{tt'}\delta_{vv'}\right]e_{t'v'}\right\|_{S_\infty^m}\\&\leq \sum_{t',v'=1}^m\left |\left(\frac{1}{D^2}\sum_{r,s=1}^Dg_{rs}^{tv}g_{rs}^{t'v'}\right)-\delta_{tt'}\delta_{vv'}\right|.
\end{align*}

Then, according to Proposition \ref{prop_gaussians} there exists a universal constant $c$ such that if $D\geq cm^4\sqrt{\log m}$, we have 
\begin{align*}
\mathbb P\left\{\max_{t,t',u, u'=1,\cdots, m}\left |\frac{1}{D^2}\sum_{r,s=1}^Dg_{rs}^{tv}g_{rs}^{t'v'}-\delta_{tt'}\delta_{vv'}\right|\geq \frac{1}{2m^4}\right\} \leq  \frac{1}{2}.
\end{align*}

Hence, the random variables $(g_{rs}^{tv})_{r,s,t,v}$ satisfy Eq. (\ref{Markov1and2}) and $$\max_{t,t',u, u'=1,\cdots, m}\left |\frac{1}{D^2}\sum_{r,s=1}^Dg_{rs}^{tv}g_{rs}^{t'v'}-\delta_{tt'}\delta_{vv'}\right|\leq \frac{1}{2m^4}$$ with positive probability.

Now, in this case we have 
\begin{align*}
\|\Phi(f_{tv})-e_{tv}\|_{S_\infty^m}&\leq  \frac{1}{2m^2},
\end{align*}so we can apply Lemma \ref{Auerbach_operator-spaces} to deduce the existence of a linear isomorphism $\psi:S_\infty^m\rightarrow S_\infty^m$ such that $\psi(e_{tv})=\Phi(f_{tv})$ for every $t,v=1,\cdots, m$ and $$\max\{\|\psi\|_{cb}, \|\psi^{-1}\|_{cb}\}\leq 2.$$

Let us consider the linear map $\eta:=\psi^{-1}\circ \Phi:S_1^D\rightarrow S_\infty^m$, which satisfies $\eta(f_{tv})=e_{tv}$ for every $t,v=1,\cdots, m$ and $$\|\eta:S_1^D\rightarrow S_\infty^m\|_{cb}\leq \|\psi:S_\infty^m\rightarrow S_\infty^m\|_{cb}\|\Phi:S_1^D\rightarrow S_\infty^m\|_{cb}\leq 2C_2\sqrt{m}.$$

Then, given the element $\tau\in \ell_2^{m^2}\otimes\cdots \otimes \ell_2^{m^2}$ ($k$-times) from  Eq. 
(\ref{element Briet-Vidick}) we consider $$z=(u\otimes \cdots \otimes u)(\tau)\in S_\infty^D\otimes \cdots\otimes S_\infty^D.$$ According to the metric mapping property of the $\epsilon$ norm and  Eq. (\ref{Markov1and2}), we can obtain the upper bound:$$\|z\|_{S_\infty^D\otimes_\epsilon \cdots\otimes_\epsilon S_\infty^D}\leq \|u:\ell_2^{m^2}\rightarrow S_\infty^D\|^k\|\tau\|_{\ell_2^{m^2}\otimes_\epsilon\cdots \otimes_\epsilon \ell_2^{m^2}}\leq \frac{\tilde{C}(k)m(\log m)^{\frac{k}{2}}}{D^k},$$where $\tilde{C}(k)$ is a constant only depending on $k$.

Finally, by the very definition of the min norm we have
\begin{align*}
\left\|\frac{1}{D^k}z\right\|_{S_1^{D}\otimes_{min}\cdots\otimes_{min}S_1^{D}}&\geq \frac{1}{D^k}\frac{1}{\|\eta:S_1^D\rightarrow S_\infty^m\|_{cb}^k}\|(\eta\otimes \cdots \otimes \eta)(z)\|_{S_\infty^m\otimes_{min}\cdots\otimes_{min}S_\infty^m}.
\end{align*}

On the other hand, it follows from the definition of $z$ and $\eta$ that $(\eta\otimes \cdots \otimes \eta)(z)=\tau$, so
\begin{align}\label{explicit quantum states}
\left\|\frac{1}{D^k}z\right\|_{S_1^{D}\otimes_{min}\cdots\otimes_{min}S_1^{D}}\geq \frac{1}{D^k}\frac{1}{\|\eta:S_1^D\rightarrow S_\infty^m\|_{cb}^k}\|\tau\|_{S_\infty^m\otimes_{min}\cdots\otimes_{min}S_\infty^m}\geq \frac{\delta(k)m^{\frac{k}{2}}}{D^k},
\end{align}where here $\delta(k)$ is a constant depending only on $k$.

Hence, we immediately deduce that $$\frac{\left\|\frac{1}{D^k}z\right\|_{S_1^{D}\otimes_{min}\cdots\otimes_{min}S_1^{D}}}{\|z\|_{S_\infty^D\otimes_\epsilon \cdots\otimes_\epsilon S_\infty^D}}\geq \frac{\frac{\tilde{\delta}(k)m^{\frac{k}{2}}}{D^k}}{\frac{\tilde{C}(k)m(\log m)^{\frac{k}{2}}}{D^k}}=C(k)\frac{m^{\frac{k}{2}-1}}{(\log m)^{\frac{k}{2}}}.$$

This concludes the proof.
\end{proof}

It is interesting to look at the elements involved in our construction. In fact, although the quantum XOR game $G$ and the measurements performed in the entangled strategy leading to the quantum bias $\beta^*(G)$ are less intuitive in our proof, the quantum state involved in the construction can be easily understood. In order to simplify notation, let us discuss the particular case $k=3$. The general case is completely analogous. Then, Eq. (\ref{explicit quantum states}) tells us that the elements $|\psi\rangle$ and $|\eta\rangle$ appearing in Lemma \ref{min on S_1} are
\begin{align*}|\psi\rangle=\sum_{i,j,k=1}^{m}g_{i,j,k}|ijk\rangle\hspace{0.3 cm}\text{and}\hspace{0.3 cm}|\eta\rangle=\sum_{i',j',k'=1}^{m}g_{i',j',k'}|i'j'k'\rangle;
\end{align*}so they are explicit realizations of independent random Gaussian variables. Of course the operator $|\psi\rangle\langle\eta|$ is not a state in the same way as the element $z$ is not a game (none of them are self-adjoint). Then, after modifying these these elements to make them self-adjoint, it is not difficult to see that the (unnormalized) state we can use in our entangled strategy is $|\phi\rangle\langle\phi|$, where $$|\phi\rangle=\sum_{i,j,k=1}^{m}g_{i,j,k}|ijk\rangle\otimes |00\rangle+\sum_{i',j',k'=1}^{m}g'_{i',j',k'}|i'j'k'\rangle\otimes |11\rangle\in \ell_{2}^{m^3}\otimes \ell_2^{2}\otimes \ell_2^{2}.$$

Hence, we are using essentially ``standard'' random states. We see once more that random quantum states present an extreme behavior in the same way as shown in many other contexts (see for instance \cite{CoNe16, FKNV, HLW06, Pal18}).

Let us finally mention that, although we have said that the game and the observables are less intuitive in our construction, one could find the precise elements that we have used and which are also defined as an explicit realizations of independent random Gaussian variables. Indeed, as we have just done to study the quantum state involved in our construction, we should first identify the corresponding mathematical objects in the proof of Theorem \ref{Key Thm} and, then, modify them to get some extra properties. For instance, it is not difficult to see that, in order to define the game $G$, we need to consider the element $$z=\sum_{r,s,r',s',r'',s''=1}^N\left(\sum_{i,j,k,i',j',k'=1}^{m}g_{i,j,k}g'_{i',j',k'}\tilde{g}_{rs}^{ii'}\tilde{g}_{r's'}^{jj'}\tilde{g}_{r''s''}^{kk'}\right)e_{rr'r''}\otimes e_{ss's''}\in S_1^{N^3},$$where the families $\{g_{i,j,k}\}_{i,j,k}$, $\{g'_{i,j,k}\}_{i,j,k}$ and $\{\tilde{g}_{rs}^{tu}\}_{r,s;t,u}$ are all independent.

Then, we should make this element self-adjoint (following for instant the same procedure we used at the beginning of the proof of Theorem \ref{Main Thm}) and, finally, we should normalize the corresponding element so that it has 1-Schatten norm equal one. This will give us an element $G$ which corresponds to a $k$- partite quantum XOR games. Note that, in order to obtain the states $\rho_x$'s and the corresponding probability distribution $p$ which define the game, we need to analyze the spectral decomposition of $G$.

\section{Acknowledgments}

We would like to thank L. Lami for many helpful discussions on LOCC operations. This work is partially based on some discussions during the online workshop ``Non-local games in quantum information theory'' from May 17 to May 21, 2021, sponsored by AIM and the NSF.

\end{document}